\documentclass{article}
\usepackage{spconf,amsmath,amssymb,epsfig,verbatim,alltt,amsthm}
\usepackage[margin=1in]{geometry}
\usepackage{tweaklist}

\renewcommand{\itemhook}{\setlength{\topsep}{5pt}%
  \setlength{\itemsep}{-1pt}}
\renewcommand{\enumhook}{\setlength{\topsep}{0.05in}%
  \setlength{\itemsep}{0.02in}}

\newtheorem{thm}{Theorem}

\title{Sparsifying Defaults: Optimal Bailout Policies for Financial Networks in Distress}

\name{Zhang Li and Ilya Pollak}
\address{School of Electrical and Computer Engineering\\
Purdue University\\
West Lafayette, IN 47906\\
li424@purdue.edu and ipollak@ecn.purdue.edu
}

\begin{document}
%\ninept
\maketitle
%\date
\pagestyle{empty}

\section{Motivation and Our Contributions}

The events of the last few years revealed an acute need for tools to systematically
model and analyze large financial networks.  Many applications of such tools include
the forecasting of systemic failures and analyzing probable effects of economic
policy decisions.

We consider the problem of optimizing the
amount and structure of a bailout in a borrower-lender network. Two broad
application scenarios motivate our work: day-to-day monitoring of financial
systems and decision making during an imminent crisis.  Examples of the latter
include the decision in September 1998 by a group of financial institutions to
rescue Long-Term Capital Management, and the decisions by the Treasury and the Fed
in September 2008 to rescue AIG and to let Lehman Brothers fail.
The deliberations leading to these and other similar actions have been extensively
covered in the press.  These reports suggest that the decision making
processes could benefit from quantitative methods for analyzing potential repercussions
of contemplated actions.  In addition, such methods could help avoid systemic
crises in the first place, by informing day-to-day actions of financial institutions
and governments.

Forecasting and preventing systemic failures is an open problem, despite
a surge in the research literature during the last four years.  There are two
main difficulties.  First, the data
on borrower-lender relationships and capital structure of financial institutions
is largely unavailable to academic researchers.  Even the
data available to regulators is far from exhaustive and perfect.  Second,
the network of financial relationships is very large, complex,
and dynamic.

Given a financial network model, we are interested in addressing the following
problem.
\begin{enumerate}
\item[]{\bf Problem I:} Given a fixed amount of cash $C$ to be injected into the system, how
should it be distributed among the nodes in order to achieve the smallest overall
amount $D$ of unpaid liabilities?
\end{enumerate}
An alternative, Lagrangian, formulation of the same problem, is to both select $C$
and determine how to distribute it in order to minimize $C+\lambda D$,
where $\lambda$ is the cost associated with every dollar of unpaid liabilities.
In this formulation, $\lambda$ can be used to model the trade-off between the costs of a bailout
(direct costs as well as moral hazard) and the costs of defaults.

%Another alternative formulation for Problem~I is
%to determine the minimum amount $C$ of cash injection and the optimal way of distributing
%it among the nodes, given the maximum acceptable overall default amount $D$.

In this work, we consider a static model with a single maturity
date, and with a known network structure.  Specifically, we assume
that we know both the amounts owed by every node in the network to
every other node, and the amounts of cash available at every node.
Even for this relatively simple model, Problem~I is far from
straightforward, because of a nonlinear relationship between the
cash injection amounts and the loan repayment amounts.  Building
upon the results from~\cite{EiNo01}, we construct algorithms for
computing exact solutions for Problem~I and its Lagrangian variant,
by showing that both formulations are equivalent to linear programs.

We also consider another problem where the objective is to minimize the number
of defaulting nodes rather than the overall amount of unpaid liabilities:
\begin{enumerate}
\item[]{\bf Problem II:} Given a fixed amount of cash $C$ to be injected into the system, how
should it be distributed among the nodes in order to minimize the number of nodes
in default, $N_d$?
\end{enumerate}
%Similar to Problem~I, we consider two alternative formulations for Problem~II: a Lagrangian
%formulation, and the determination of the minimum
%cash injection amount for a given maximum acceptable
%number of defaults.

For Problem~II, we develop an approximate algorithm using a reweighted $\ell_1$ minimization
approach inspired by~\cite{CaWaBo08}.  We illustrate our algorithm using an
example with synthetic data for which the optimal solution can be calculated exactly,
and show through numerical simulation that the solutions calculated by our algorithm
are close to optimal.

In Section~\ref{sect:model} we describe our model and the results from prior literature
that we use.  Our own results---the equivalence of Problem~I to a linear program
and the approximate algorithm for Problem~II---are described in Section~\ref{sect:results}.
%As indicated above, there is a considerable gap between our model and situations typically
%arising in practice.  Section~\ref{sect:future} discusses future research directions aimed
%at bridging this gap.

\section{Notation, Model, and Background}
\label{sect:model}

\begin{table}
\caption{Notation for several vector \vspace*{-0.05in} quantities}
%\centering
\begin{center}
\begin{tabular}{|l|p{2in}|}
\hline
\textsc{vector} & \textsc{$i$-th component} \\
\hline
{\bf 0} & 0\\
\hline
{\bf 1} & 1\\
\hline
${\bf e} \geq {\bf 0}$ & cash on hand at node $i$ \\
\hline
${\bf c} \geq {\bf 0}$ & external cash injection to node $i$ \\
\hline
$\bar{\bf p}$ & the amount node $i$ owes to all its creditors \\
\hline
${\bf p} \leq \bar{\bf p}$ & the total amount node $i$
      actually repays all its creditors on the due date of the loans \\
\hline
$\bar{\bf p} - {\bf p}$ & node $i$'s total unpaid liabilities \\
\hline
${\bf q}$ & the total amount node $i$ actually receives from all its borrowers \\
\hline
${\bf r} = {\bf q} + {\bf e} + {\bf c}$ & the total funds available to $i$ for making
payments to its creditors \\
\hline
\end{tabular}
\vspace*{-0.2in}\\
\end{center}
%{\small Can put some text here if necessary.}
\label{tb:notation}
\end{table}

Our network model is a directed graph with $N$ nodes where a
directed edge from node $i$ to node $j$ with weight $L_{ij}>0$
signifies that $i$ owes $\$L_{ij}$ to $j$. This is a one-period
model with no dynamics---i.e., we assume that all the loans are due
on the same date and all the payments occur on that date. We use the
following notation:
\begin{itemize}
\item any inequality whose both sides are vectors is component-wise;
\item $\mathbf{0}$, $\mathbf{1}$, ${\bf e}$, ${\bf c}$, $\bar{\bf p}$, ${\bf p}$, ${\bf q}$, and ${\bf r}$
are all vectors in $\mathbb{R}^N$ defined in Table~\ref{tb:notation};
\item $D = \mathbf{1}^T(\bar{\bf p} - {\bf p})$ is the overall amount of unpaid liabilities in
the system;
\item $N_d$ is the number of nodes in default, i.e., the number of nodes $i$
whose payments are below their liabilities, $p_i < \bar{p}_i$;
\item $\Pi_{ij}$ is what node $i$ owes to node $j$, as a fraction of the total
amount owed by node $i$,
\[
\Pi_{ij} = \left\{\begin{array}{ll} \frac{L_{ij}}{\bar{p}_i} & \mbox{if } \bar{p}_i \neq 0, \\
0 & \mbox{otherwise;} \end{array}\right.
\]
\item $\Pi$ and $L$ are the matrices whose entries are $\Pi_{ij}$ and $L_{ij}$, respectively.
\end{itemize}
Following~\cite{EiNo01}, we make the following assumptions.
\begin{itemize}
\item If $i$'s total funds are at least as large as its liabilities (i.e.,
$r_i \geq \bar{p}_i$) then all $i$'s creditors get paid in full.
\item All $i$'s debts have the same seniority.  This means that, if $i$'s
liabilities exceed its total funds (i.e., $r_i < \bar{p}_i$) then each
creditor gets paid in proportion to what it is owed.
This guarantees that the amount actually received by node $j$ from node $i$ is
always $\Pi_{ij} p_i$.  Therefore, the total amount received by any node $i$ from
all its creditors is $q_i = \sum_{j=1}^N \Pi_{ji} p_j$.
\end{itemize}
As defined in~\cite{EiNo01}, a {\em clearing payment vector}
${\bf p}$ is a vector of borrower-to-lender payments that is consistent
with these conditions
for given $L$, ${\bf e}$, and ${\bf c}$.  It is shown
in~\cite{EiNo01} (Theorem 2) that a unique ${\bf p}$ exists
for any network that satisfies a mild technical assumption.
We restrict our attention to models that satisfy this assumption
and therefore have a unique clearing payment vector ${\bf p}$.

\section{Results}
\label{sect:results}

\subsection{Minimizing the amount of unpaid liabilities}

Consider a network with a known structure of liabilities $L$
and a known cash vector ${\bf e}$.
Using the notation established in the preceding section,
we can see that Problem~I seeks to find a cash injection allocation
vector ${\bf c}$ to minimize the total amount of unpaid liabilities,
\[
D = \mathbf{1}^T(\bar{\bf p} - {\bf p}),
\]
subject to the constraint that the total amount of cash injection
is some given number $C$:
\[
\mathbf{1}^T{\bf c} = C.
\]
Our first result establishes the equivalence of Problem~I and
a linear programming problem.

\begin{thm}
\label{thm:LP} Assume that the liabilities matrix $L$, the
cash-on-hand vector ${\bf e}$, and the total cash injection amount
$C$ are fixed and known. Assume that the network satisfies all the
conditions listed above.  Then Problem~I has a solution which can be
obtained by solving the following linear program:
\begin{align}
& \mbox{find } {\bf c} \mbox{ and } {\bf p}
\mbox{ to maximize } \mathbf{1}^T{\bf p}
\label{eq:LP}
\\
& \mbox{subject to } \nonumber\\
& \mathbf{1}^T{\bf c} = C, \nonumber\\
& {\bf c} \geq \mathbf{0}, \nonumber\\
& \mathbf{0}\leq {\bf p} \leq \bar{\bf p}, \nonumber\\
& {\bf p} \leq \Pi^T{\bf p} + {\bf e} + {\bf c}. \nonumber
\end{align}
\end{thm}

\begin{proof}
Since the constraints on ${\bf c}$ and ${\bf p}$ in linear
program (\ref{eq:LP}) form a closed and bounded set in $\mathbb{R}^{2N}$,
a solution exists.  Moreover,
for any fixed ${\bf c}$, it follows from Lemma 4 in~\cite{EiNo01}
that the linear program has a unique solution for ${\bf p}$ which
is the clearing payment vector for the system.

Let $({\bf p}^\ast,{\bf c}^\ast)$ be a solution to (\ref{eq:LP}).
Suppose that there exists a cash injection allocation that leads
to a smaller total amount of unpaid liabilities than does ${\bf c}^\ast$.
In other words, suppose that there exists ${\bf c'}>\mathbf{0}$,
with $\mathbf{1}^T{\bf c'} = C$, such that the
corresponding clearing payment vector ${\bf p'}$ satisfies
$
\mathbf{1}^T(\bar{\bf p} - {\bf p'}) <
\mathbf{1}^T(\bar{\bf p} - {\bf p}^\ast),
$
or, equivalently,
\begin{align}
\mathbf{1}^T{\bf p}^\ast < \mathbf{1}^T{\bf p'}.
\label{eq:contradiction}
\end{align}
Note that ${\bf c'}$ satisfies the first two constraints of
(\ref{eq:LP}).  Moreover, since ${\bf p'}$ is the corresponding
clearing payment vector, the last two constraints are satisfied
as well.  The pair $({\bf p'},{\bf c'})$ is thus in the
constraint set of our linear program.  Therefore,
Eq.~(\ref{eq:contradiction}) contradicts the assumption that
$({\bf p}^\ast,{\bf c}^\ast)$ is a solution to (\ref{eq:LP}).
This completes the proof that ${\bf c}^\ast$ is the allocation of
$C$ that achieves the smallest possible amount $D$ of unpaid
liabilities.
\end{proof}

In the Lagrangian formulation of Problem~I, we are given a weight
$\lambda$ and must choose the total cash injection
amount $C$ and its allocation ${\bf c}$ to
minimize $C+\lambda D$. This is equivalent to the following
linear program:
\begin{align}
& \mbox{find } C, {\bf c}, \mbox{ and } {\bf p}
\mbox{ to maximize } \lambda\mathbf{1}^T{\bf p} - C
\label{eq:LP2}
\\
& \mbox{subject to the same constraints as in (\ref{eq:LP}).} \nonumber
\end{align}
This equivalence follows from Theorem~\ref{thm:LP}: denoting
a solution to (\ref{eq:LP2}) by $(C^\ast,{\bf p}^\ast,{\bf c}^\ast)$,
we see that the pair $({\bf p}^\ast,{\bf c}^\ast)$ must be a solution
to (\ref{eq:LP}) for $C=C^\ast$.  At the same time, the fact that
$C^\ast$ maximizes the objective function in (\ref{eq:LP2}) means that
it minimizes $C+\lambda D = C + \lambda\mathbf{1}^T(\bar{\bf p}-{\bf p})$,
since $\bar{\bf p}$ is a fixed constant.

\subsection{Minimizing the number of defaults}

Given that the total amount of cash injection
is $C$, Problem~II seeks to find a cash injection allocation
vector ${\bf c}$ to minimize the number of defaults $N_d$,
i.e., the number of nonzero entries in the vector
$\bar{\bf p} - {\bf p}$.

We adapt the reweighted $\ell_1$ minimization strategy
approach from Section 2.2 of~\cite{CaWaBo08}.  Our algorithm
solves a sequence of weighted versions of the linear program (\ref{eq:LP}),
with the weights designed to encourage sparsity of $\bar{\bf p} - {\bf p}$.
In the following pseudocode of our algorithm,
${\bf w}^{(m)}$ is the weight vector during the $m$-th iteration.
\begin{enumerate}
\item $m \leftarrow 0$.
\item Select ${\bf w}^0$ (e.g., ${\bf w}^0 \leftarrow {\bf 1}$).
\item Solve linear program (\ref{eq:LP}) with objective function replaced by
${\bf p}^T{\bf w}^{(m)}$.
\item Update the weights: for each $i=1,\cdots,N$,
\begin{equation}
w_i^{(m+1)}\leftarrow \frac{K}{\exp\left(\bar{p}_i-p^{\ast(m)}_i\right)+\epsilon},
\nonumber
\end{equation}
where $K>0$ and $\epsilon>0$ are constants, and ${\bf
p}^{\ast(m)}$ is the clearing payment vector obtained in Step 3.
\item If $\|{\bf w}^{(m+1)}-{\bf w}^{(m)}\|_1 < \delta$, where $\delta>0$ is a constant, stop;
else, increment $m$ and go to Step 3.
\end{enumerate}

%\begin{figure}[!ht]
\begin{figure}
    \centering
    \includegraphics[width=0.47\textwidth]{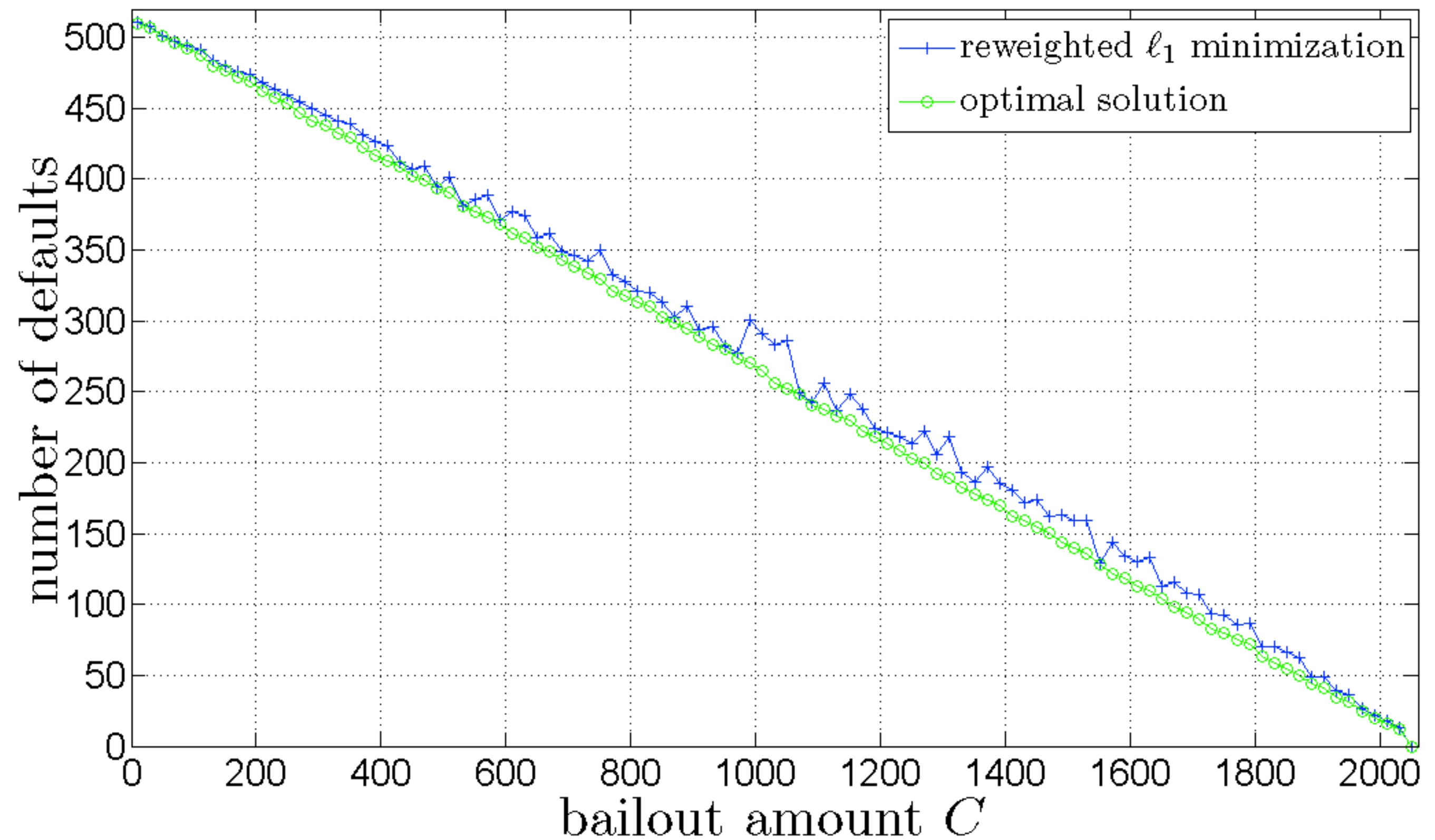}
    \caption{Illustration of the algorithm  for Problem~II.}
\label{fig:evaluation}
\end{figure}

\vspace*{-0.2in}
We test the algorithm on a network for which we know the optimal
solution. We use a full binary tree with 10 levels and $N=2^{10}-1$ nodes:
levels 0 and 9 correspond to the root and the leaves, respectively.  Every node at level $s<9$
owes $\$2^{10-s}$ to each of its two creditors (children).  We set ${\bf e} = {\bf 0}$.

If $C=0$,
then all 511 non-leaf nodes are in default, and the 512 leaves are not
in default.  In aggregate, the nodes at any level $s<9$
owe $\$2048$ the nodes at level $s+1$.  Therefore, if
$C\geq \$2048$, then $N_d=0$ can be achieved by allocating the entire amount
to the root node.

For $0 < C < 2048$, we first observe
that if $C = 2^{11-s}$ for some integer $s$, then the optimal solution is to allocate
the entire amount to a node at level~$s$.  This would prevent the defaults of this node and all
its non-leaf descendants, leading to $511-(2^{9-s}-1)$ defaults.  If $C$ is not a power of two,
we can represent it as a sum of powers of two and apply the same argument recursively,
to yield the following optimal number of defaults:
\[
N_d=\displaystyle 511-\sum_{u=3}^U b(u)\cdot (2^{u-2}-1),
\]
where $b(u)$ is the $u$-th bit in the binary representation
of $C$ (right to left) and $U$ is the number of bits.
The green line in Fig.~\ref{fig:evaluation} is a plot of
the minimum number of defaults as a function of $C$.  The
blue line is the solution calculated by our reweighted
$\ell_1$-minimization algorithm
with $K=1000$, $\epsilon=0.001$ and $\delta=0.001$.
The algorithm was run using six different initializations:
five random ones and $\bf w^{(0)} = 1$.  Among the six solutions,
the one with the smallest number of defaults was selected.
As evident from Fig.~\ref{fig:evaluation}, the results are
very close to the optimal for the entire range of $C$.

\bibliographystyle{plain}
{\footnotesize
\bibliography{WIDS}

\begin{thebibliography}{1}

\bibitem{CaWaBo08}
E.J. Cand\`{e}s, M.B. Wakin, and S.P. Boyd.
\newblock Enhancing sparsity by reweighted $\ell_1$ minimization.
\newblock {\em Journal of Fourier Analysis and Applications}, 14(5-6):877--905,
  2008.

\bibitem{EiNo01}
L.~Eisenberg and T.H. Noe.
\newblock Systemic risk in financial systems.
\newblock {\em Management Science}, 47(2):236--249, February 2001.

\end{thebibliography}
}
\end{document}